\documentclass[11pt]{article}

\usepackage[left=1in,right=1in,top=1in,bottom=1in]{geometry}

\usepackage{charter}
\usepackage[charter,cal=cmcal]{mathdesign}
\usepackage{microtype}

\usepackage{color}
\definecolor{darkblue}{rgb}{0,0,0.5}
\definecolor{darkgreen}{rgb}{0,0.5,0}
\usepackage[colorlinks=true,linkcolor=darkblue,urlcolor=darkblue,citecolor=darkgreen]{hyperref}

\usepackage{graphicx}

\usepackage{amsmath}
\newtheorem{theorem}{Theorem}[section]
\newtheorem{definition}[theorem]{Definition}
\newenvironment{proof}{\noindent \textsc{Proof.} }{\hfill $\square$\vspace{10pt}}

\input{Qcircuit}
\newcommand{\tensor}{\otimes}
\newcommand{\ketbra}[2]{\ket{#1} \bra{#2}}
\newcommand{\braket}[2]{\langle #1 | #2 \rangle}

\newcommand{\poly}{\ensuremath{\mathrm{poly}}}
\usepackage{multirow}
\usepackage{algorithm}

\title{\bf Quantum Coins}

\author{Michele Mosca$^{1}$ and Douglas Stebila$^{2}$ \\%
\ \\
\it \small $^{1}$ Institute for Quantum Computing and Department of Combinatorics, University of Waterloo \\
\it \small Perimeter Institute for Theoretical Physics, Waterloo, Ontario, Canada \\
\it \small $^{2}$ Information Security Institute, Queensland University of Technology, Brisbane, Queensland, Australia \\
\small Email: \href{mailto:mmosca@iqc.ca}{\tt mmosca@iqc.ca}, \href{mailto:douglas@stebila.ca}{\tt douglas@stebila.ca}
}

\date{November 7, 2009}

\begin{document}

\maketitle

\begin{abstract}
One of the earliest cryptographic applications of quantum information was to create quantum digital cash that could not be counterfeited.  In this paper, we describe a new type of quantum money: \textbf{quantum coins}, where all coins of the same denomination are represented by identical quantum states.  We state desirable security properties such as anonymity and unforgeability and propose two candidate quantum coin schemes: one using black box operations, and another using blind quantum computation.

\end{abstract}

\section{Introduction}\label{sec:qu:qm:intro}

The uncertainty principle and no-cloning theorem of quantum mechanics made quantum money one of the original interests of quantum information theory.  The ability to create digital money which cannot be counterfeited because of the laws of physics is a compelling idea.  Classical digital cash has been researched extensively, with ongoing improvements to its security tradeoffs, but remains fundamentally subject to the constraint that classical bits can be easily copied.  With quantum money, we hope to use the inability to perfectly clone quantum states to prevent counterfeiting.  Besides being non-counterfeitable, an effective digital cash scheme should also be efficiently verifiable, anonymous, transferable, and robust.  

In this paper, we describe a new form of quantum money called \emph{quantum coins}, where all coins of the same denomination are represented by identical quantum states.  We state formally what it means for them to be unforgeable and describe how to implement quantum coin schemes using black box operations and using blind quantum computing.  We also describe \emph{quantum bills} which capture a wide range of notions of quantum money.

\paragraph{Contributions.}
In this paper, we present a new type of quantum money, which we call \emph{quantum coins}: coins are transferable, locally verifiable, and unforgeable, and have some anonymity properties.  Each coin generated by the bank should be a copy of the same quantum state, and hence coins should be indistinguishable from one another.  Additionally, a circuit is provided to allow the coins to be verified locally and then transferred for later use.  

We describe how to achieve quantum coins with black box quantum circuits and with blind quantum computation.  The unforgeability of coins in our scheme comes from complexity theoretic assumptions on the adversary's running time.

Our work contrasts with previous quantum money schemes, which we call \emph{quantum bills}: in a quantum bill scheme, the bank generates tokens that are classical/quantum pairs, which in general are distinct.  The classical string may serve as a serial number or as some input value to be used in the verification procedure.  

\paragraph{Future directions.}
Our quantum coin construction of Section~\ref{sec:qu:qm:bbcoins} requires the use of a black-box oracle in the verification circuit, but it is not yet known how these can be implemented.  An open question is to find a way to obfuscate the verification circuit so that it is effectively a black box, and in general to find a model for obfuscation of quantum circuits, possibly using computational assumptions.  We describe how blind quantum computation could be used in the context of quantum coin verification and note the limitations, in particular the online quantum communication required.  Reducing the communication and computational requirements of blind quantum computing is a problem that merits further study.  

Although our coins are inherently anonymous if the bank issues coins correctly, we do not yet have a mechanism to allow users of the system to verify that the coins are indeed issued correctly, so this remains an open question.

In Section~\ref{sec:qu:qm:types:bills}, we briefly discuss a model for quantum bills.  An open question related to quantum bills is to find an offline-verifiable quantum bill scheme; this may require using computational hardness assumptions.

\paragraph{Outline.} The remainder of the paper is organized as follows.  In Section~\ref{sec:qu:qm:goals}, we describe the goals for a quantum money scheme and analyze existing quantum money schemes, as well as our own, in relation to these goals.  Section~\ref{sec:qu:qm:types} introduces the two main types of quantum money, quantum coins and quantum bills, and describes their precise security properties.  In Section~\ref{sec:qu:qm:bbcoins}, we describe how to implement quantum coins in the black box model and give bounds on unforgeability.  In Section~\ref{sec:qu:qm:blindcoins}, we discuss implementing quantum coins using blind quantum computation.

\subsection{Related work}

\paragraph{Digital cash.}
Digital cash has been well-explored in classical cryptographic contexts, with the first schemes being proposed by Chaum \cite{Cha85,Cha88} and Chaum, Fiat, and Naor \cite{CFN88}.  For classical digital cash schemes, one of the main problems to solve is the \emph{multiple-spending problem}: since classical digital cash can easily be duplicated, there must be a way to prevent the same tokens from being redeemed more than once.  An online scheme, in which each token is verified with the bank at the time it is meant to be spent, solves this problem immediately, but online verification requires an online communications channel between merchant and bank.  The other general solution for preventing multiple spending is to embed some identity information in the money tokens such that, if the token is spent only once, the transaction remains anonymous, but if the token is spent multiple times, then the bank can combine these multiple transactions to recover the identity of the multiple spender.  Moreover, classical digital cash is not transferable unless we allow the size of the token to grow linearly in the number of transfers \cite{CP92}.  

\paragraph{Quantum money.}
Quantum money was one of the earliest applications of quantum information theory, and was introduced in the early papers of Wiesner \cite{Wie83} and Bennett, Brassard, Breidbard, and Wiesner \cite{BBBW82}.  In both schemes, a bank constructs distinct quantum tokens and corresponding classical serial numbers.  The tokens are the encoding of a random string in randomly chosen basis states of two non-orthogonal bases; the no-cloning theorem prevents perfect cloning of individual tokens.  However, the tokens can only be verified by the bank: verification requires knowledge of the bases chosen for each token and the classical string that should be obtained upon measurement in the appropriate bases.  This means that an online quantum channel is required between merchants and the bank.  The tokens are non-transferable and are not anonymous.

Tokunaga, Okamoto, and Imoto \cite{TOI03} give a scheme for non-transferable anonymous quantum cash with online verification.  In their scheme, a user obtains a distinct token from the bank; tokens are generated using private parameters and random values stored by the bank.  The user then alters the token with an appropriate randomly chosen unitary transformation to obtain anonymity.  At payment time, the user presents the token to the merchant who transmits it (over a quantum channel) to the bank for verification.  The scheme is secure against an attacker who can examine a single token, but has not been proven secure against an attacker who can obtain and examine all the quantum tokens.

Our work on quantum coins makes use of work by Aaronson \cite{Aar05} that introduced a complexity-theoretic no-cloning theorem that allows us to argue for the unforgeability of quantum coins.  Our work was first presented in \cite{MS06}, \cite{MS07}, and \cite{Ste09}.  Subsequently Aarsonson expanded his work based on discussions with us to also include a presentation of quantum money \cite{Aar09} similar to ours; we have noted in footnotes throughout this paper where that he presents similar concepts.

\section{Security goals}\label{sec:qu:qm:goals}

We now describe, informally, the properties that a good money scheme should have.

\begin{enumerate}
\item[G1.] \emph{Anonymous}: it should be difficult for any party to trace the use of a token to determine who spent it or where they spent it.
\item[G2.] \emph{Unforgeable}: given zero or more tokens and the verification circuit, it should be difficult for a forger to produce another token that passes the verification procedure with non-negligible probability.
\item[G3.] \emph{Efficiently locally verifiable}: there should be an efficient algorithm that can determine with high accuracy whether a token is valid or not, without communicating with the bank.
\item[G4.] \emph{Transferable}: a valid token should be unchanged by the verification procedure, and thus can be transferred and reused in a subsequent verification procedure.
\end{enumerate}

We will formally define unforgeability for quantum coin schemes in Section \ref{sec:qu:qm:types:coins:forgery}.

Figure~\ref{fig:goals:comparison} shows which of the above goals are satisfied by various existing money schemes.  The ``type'' column indicates whether the tokens for a given denomination are all identical (``coin'') or different (``bill'').
For classical digital cash schemes, we note that while unforgeability is impossible, it is possible to detect double spending of a token and trace it back to the offending party; such schemes, however, offer anonymity and offline double-spending detection only with computational assumptions.  Our quantum coin schemes offer ``partial'' anonymity as we describe in Section~\ref{sec:qu:qm:types:coins:anon}.  Additionally, the size of transferable digital cash must grow linearly in the number of transfers  \cite{CP92}.

\begin{figure}[ht]
\begin{center}
{
\begin{tabular}{l|c|c|c|c|c}
& & & & \multirow{3}{1in}{\centering \bf Efficiently\\locally\\verifiable} & \\
{\bf Scheme} & {\bf Type} & {\bf Anonymous} & {\bf Unforgeable} & & {\bf Transferable} \\
& & & & & \\ \hline \hline
Physical coins & coin & yes & physically & yes & yes \\ \hline
Physical bills & bill & no & physically & yes & yes \\ \hline
Classical digital & bill & yes & double-spending & yes & grows in \\
cash & & & detection & & size \\ \hline
\cite{Wie83} & q. bill & no & yes & no & no \\ \hline
\cite{BBBW82} & q. bill & no & yes & no & no \\ \hline
\cite{TOI03} & q. bill & yes & yes & no & no \\ \hline
This work: & q. coin & partially & yes & yes & yes \\
black box & & & & & \\ \hline
This work: & q. coin & partially & yes & no & yes \\
blind computation & & & & & \\ \hline
\end{tabular}
}
\end{center}
\caption{Summary of money schemes and their properties}\label{fig:goals:comparison}
\end{figure}

\section{Types of quantum money}\label{sec:qu:qm:types}

\subsection{Quantum coins}\label{sec:qu:qm:types:coins}

In one type of quantum money, \emph{quantum coins}, a bank issues many tokens for a particular denomination, and all these tokens are (supposed to be) copies of the same quantum state.  The state for a 5-cent coin, for example, might be the pure state $\ket{\psi_{5}}$ and the bank produces many copies $\ket{\psi_{5}}^{\otimes 1000000}$, issuing one copy to each person who withdraws 5 cents from the bank.  We use the term \emph{quantum coin} because physical coins in the real world have the same property: there should be no discernible difference between different coins of the same denomination.  The specification of a quantum coin scheme consists of the specification of the money state and the verification circuit.

\begin{definition}\label{defn:qu:qm:types:coins}
A \emph{quantum coin scheme} is a pair $(V, \ket{\psi})$, where $\ket{\psi}$ is an $n$-qubit pure state in a $2^{n}$-dimensional Hilbert space $\mathcal{H}^{2^{n}}$, and $V$ is a quantum circuit with a quantum $n$-qubit input register (denoted $\rho$), plus optional ancilla quantum registers, a classical output bit, and a quantum output register of $n$ qubits.
\end{definition}

The basic scenario of how a quantum coin scheme would operate is as follows.  A bank generates a large number of quantum coins and stores them.  A user withdraws coins from the bank via a private quantum channel and stores the coins.  When the user wishes to spend the coins, it transfers the coins to the merchant using a quantum channel.  The merchant uses a quantum circuit to verify the coins; this procedure may or may not involve classical or quantum communication with the bank.  Finally, the merchant stores the coins until redeeming them with the bank or issuing them as change to subsequent users.

\subsubsection{Verification}\label{sec:qu:qm:types:coins:verif}

In the most general setting, the verification circuit $V$ operates on three registers: a 1-qubit data readout register, an $n$-qubit input register, and an arbitrary $m$-qubit ancilla.  After applying $V$, the first register is measured, and the output is the decision on whether to accept the token as valid or not.  If the input is a valid quantum coin $\ket{\psi}$, then, after the application of $V$ and the measurement, the classical output should be 0 and the partial trace over the first and third registers should leave the second register in the same state $\ket{\psi}$.  The circuit diagram is given in Figure~\ref{fig:types:coins:verif-circ-general}.

\begin{figure}[ht]
\[
\Qcircuit @C=1em @R=.7em {
\lstick{\ket{0}} & \qw & \multigate{2}{V} & \meter & \cw \\
\lstick{\rho} & {/} \qw & \ghost{V} & {/} \qw & \qw \\
\lstick{\ket{0}^{\otimes m}} & {/} \qw & \ghost{V} & &
}
\]
\caption{Generic verification circuit for a quantum coin scheme $(V, \ket{\psi})$.}
\label{fig:types:coins:verif-circ-general}
\end{figure}
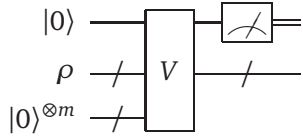

We cannot simply provide this circuit in an unprotected form to the public: it may be possible to decompose the circuit into component gates and find a way to forge money.  In Section~\ref{sec:qu:qm:bbcoins} we describe two techniques for implementing this circuit in a safe way: (1) black box verification, in which we assume the circuit is a black box and security rests on complexity-theoretic assumptions, and (2) blind quantum computation, which allows one party to implement an operation without gaining any information about the operation being performed, and security is information-theoretic.  It could be possible to construct a scheme based on computational assumptions.

\subsubsection{Unforgeability}\label{sec:qu:qm:types:coins:forgery}

We assume that a forger has the verification circuit $V$ and many (or all) tokens issued, say $k$ of them.  The goal of a forger is to produce a state that passes more than $k$ verification tests with good probability.  Since the verification circuit projects the state into the subspace spanned by $\ket{\psi}$, this is equivalent to creating a state that has good overlap with the state $\ket{\psi}^{\otimes k+1}$.

\begin{definition}\label{defn:qu:qm:types:coins:unforgeable}
A quantum coin scheme $(V, \ket{\psi})$, where $\ket{\psi}$ is an $n$-qubit state, is \emph{unforgeable} if, given the verification circuit $V$ and $k$ copies of the state $\ket{\psi}$, for any $k \ge 0$, $k \in \poly(n)$, it is not possible for a quantum adversary running in time $\poly(n)$ to produce a state $\rho$ such that $\bra{\psi}^{\otimes k+1} \rho \ket{\psi}^{\otimes k+1}$ is non-negligible (in $n$).\footnote{In the language of Aaronson \cite{Aar09}, this is a single key public key quantum money scheme with completeness error 0 and soundness error negligible in $n$.}
\end{definition}

In order to prevent a counterfeiter from performing quantum state tomography \cite{AJK04} and precisely determining the state $\ket{\psi}$, the bank should avoid issuing more than a polynomial number (in $n$) of coins.

Information theoretically, no offline quantum coin scheme can be perfectly unforgeable (that is, with $\bra{\psi}^{\otimes k+1} \rho \ket{\psi}^{\otimes k+1}=0$ and no running time restriction in Definition~\ref{defn:qu:qm:types:coins:unforgeable}).  If a forger has a verification circuit and unbounded quantum computational resources, the forger can repeatedly generate test states until one such state passes; after verification, this state is projected into a valid money state and can subsequently be used as a money token.  Thus, we must introduce computational assumptions on a forger and attempt to lower bound the amount of work required to forge.

Without any further specification of the quantum coin scheme and the verification circuit, we cannot say anything more about the unforgeability of such schemes.  In Section~\ref{sec:qu:qm:bbcoins:forgery}, we show that a black box quantum coin scheme is unforgeable.

\subsubsection{Anonymity}\label{sec:qu:qm:types:coins:anon}

In our ideal formulation, all quantum coins (for a particular denomination) are minted as the same quantum state $\ket{\psi}$.  However, the bank could create quantum coins from different quantum states, all of which can be verified by a particular verification circuit.  Although we have no procedure for users to test the anonymity of the system, it would be possible for a regulator to regularly review the procedures of the bank and ensure that it is issuing identical tokens as the coins.  If indeed all the coins issued are identical, then it is impossible for the use of a coin to be tracked.  If quantum circuits can be obfuscated, then the verification circuit could be provided in an obfuscated form as a fixed public classical string which merchants then implement; since the circuit is fixed for all merchants, this would give anonymity to merchants as well.  If an interactive protocol is required for verification (as in our use of blind quantum computing in Section~\ref{sec:qu:qm:blindcoins}), then anonymous classical \cite{BT07} and quantum \cite{BBFGT07} communication can be used to improve the anonymity of merchants.

\subsection{Quantum bills}\label{sec:qu:qm:types:bills}

Whereas all quantum coins of the same denomination are identical states, with \emph{quantum bills} we allow tokens of the same denomination to be different quantum states and additionally allow some classical information associated with each quantum state.  So a bank might issue a set of states $\{ (s_{i}, \ket{\psi_{i}}) : i \in \Gamma \}$ as the valid \$20 bills.  This corresponds to physical bills which have a distinct serial number on each bill.  

An example of an approach one might take to making quantum bills would be the following. Let $a$ be an element of order $m$ of some group $G$ and let $r$ be a function that encrypts elements of $G$. Suppose there were a way to publish a circuit $C$ that implements, for any group element $b$ and integer $y \in \{0,1,\ldots, m-1\}$, the mapping $\ket{y}\ket{r(b)} \rightarrow \ket{y} \ket{r(b a^{y})}$ but from which one cannot (among other things) determine $x$ given $\ket{r(a^x)}$. (Note that the standard quantum discrete logarithm algorithm for computing $x$ would require a means for computing $r(a^{zx+y})$ for arbitrary integers $z$ and $y$.) Then a possible way to generate quantum money is for a bank to perform eigenvalue estimation (starting from a state $\ket{r(b)}$) in order to generate a random eigenstate of the operation induced by $C$, of the form 
\[ \ket{\psi_{k}} = \sum_{x=0}^{m-1} e^{-2 \pi i kx/m} \ket{r(b a^x)} \enspace , \]
together with the eigenvalue parameter $k$. The bank would publish an authentic list of valid parameters $k$. The bill would consist of the state $\ket{\psi_k}$ and the classical value $k$, which any verifier could check by performing eigenvalue estimation on the bill and confirming the eigenvalue parameter is $k$ (and that $k$ is on the authentic list of valid serial numbers). There are many variations of this approach that one might try, and many open questions. We will focus on quantum coins in this paper.

\begin{definition}\label{defn:qu:qm:types:bills}
A \emph{quantum bill scheme} is a pair $(V, \{ (s_{i}, \ket{\psi_{i}}) : i \in \Gamma \})$, where $\Gamma$ is a finite set, and for each $i \in \Gamma$, $s_{i}$ is a label in a set $\mathcal{S}$, $\ket{\psi_{i}}$ is an $n$-qubit pure state in a $2^{n}$-dimensional Hilbert space $\mathcal{H}^{2^{n}}$.  Moreover, $V$ is a quantum circuit with a quantum input register (denoted $\ket{s}$), a quantum $n$-qubit input register (denoted $\rho$), plus optional ancilla quantum registers, a classical output bit, and a quantum output register of $n$ qubits.\footnote{In the language of Aaronson \cite{Aar09}, this is a public key quantum money scheme.}
\end{definition}

\subsubsection{Verification}\label{sec:qu:qm:types:bills:verif}

A generic verification circuit for a quantum bill scheme is given in Figure~\ref{fig:types:bills:verif-circ-general}.

\begin{figure}[ht]
\[
\Qcircuit @C=1em @R=.7em {
\lstick{\ket{0}} & \qw & \multigate{3}{V} & \meter & \cw \\
\lstick{\ket{s}} & {/} \qw & \ghost{V} & & \\
\lstick{\rho} & {/} \qw & \ghost{V} & {/} \qw & \qw \\
\lstick{\ket{0}^{\otimes m}} & {/} \qw & \ghost{V} & &
}
\]
\caption{Generic verification circuit for a quantum bill scheme $(V, \{ (s_{i}, \ket{\psi_{i}}) : i \in \Gamma \}$).}
\label{fig:types:bills:verif-circ-general}
\end{figure}
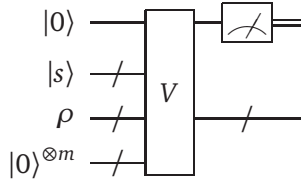

The use of the classical label $s_{i}$ may vary according to the scheme.  For example, in the schemes of Wiesner \cite{Wie83} and Bennett {\it et al.} \cite{BBBW82}, $s_{i}$ is a serial number that allows the issuer to retrieve the verification details, while in the scheme of Tokunga {\it et al.} \cite{TOI03}, $s_{i}$ is effectively unused; in their scheme it is used to represent the denomination of the bill (e.g., \$5), but in our formulation the denomination is fixed for a particular scheme so the label is effectively the empty string for all $i \in \Gamma$.  Schemes where $s_{i}$ is non-trivial and unchanged by verification inherently limit the anonymity of the scheme, just as the serial number on physical bills places some limits on anonymity.

While all previous quantum money schemes discussed in Section~\ref{sec:qu:qm:intro} are classified as quantum bill schemes based on the above definition, none of them satisfy all of the security properties described in Section~\ref{sec:qu:qm:goals}.  In particular, no previous quantum money scheme is offline verifiable: all previous schemes require that the issuer verify a token via quantum communication, a requirement which we aim to remove for quantum coins.  In the rest of this paper, we are only concerned with quantum coin schemes, not quantum bill schemes.

\section{Black box quantum coins}\label{sec:qu:qm:bbcoins}

Our first implementation for quantum coins works in the black box circuit model. We assume the verification circuit provided to the public is a black box: ``anything one can compute from it one could also compute from the input-output behavior of the program'' \cite[p. 2]{BGIRSVY01full}.  With this assumption, we present a scheme in which coins are unforgeable.  The scheme allows coins to be transferred an arbitrary number of times.  The use of a black box circuit means that coins can be verified locally without any communication, classical or quantum, with the bank.

We note that it is not known at present whether a quantum circuit can be implemented as a true black box.  There are pessimistic results about the ability to obfuscate classical circuits \cite{BGIRSVY01}, although loopholes do exist: for example, point functions can be obfuscated \cite{Wee05}.  However, no results are known about quantum circuits.  Another classical technique for black box computation is physically tamper-proof hardware, but again the parallel in quantum computation is not clear.

In our black box construction, a coin is a randomly chosen secret state, and the verification circuit recognizes precisely that state using an oracle like the iterate in amplitude amplification \cite{BBHT98}.

Let $\ket{\psi}$ a pure state chosen randomly (according to the Haar measure) from among the pure states in $\mathcal{H}^{2^{n}}$.  The verification oracle is $U_{\psi}=I-2\ketbra{\psi}{\psi}$.  Since this is a black-box oracle scheme, the unforgeability proof of Section~\ref{sec:qu:qm:bbcoins:forgery} applies and the scheme is unforgeable in the black-box oracle model.  

In practice, however, choosing a pure state $\ket{\psi}$ randomly according to the Haar measure with the additional constraints that we must be able to compute $I-2\ket{\psi}\bra{\psi}$ and that we must be able to produce many copies of $\ket{\psi}$ is problematic and it is not known how to do so in polynomial time.  Recent work has focused on developing \emph{approximate quantum $t$-designs} \cite{AE07} where, roughly speaking, $t$ copies of a state can be efficiently constructed such that tensor product state is sufficiently close to $t$ copies of a state selected uniformly at random according to the Haar measure.  Aaronson \cite[Theorem~8]{Aar09} gives a technique for constructing $t \in \poly(n)$ copies of a pseudorandom state that are nearly indistinguishable (that is, negligibly different) from $t$ copies of a truly random state by any measurement, even allowing the measurement procedure to make $\poly(n)$ calls to an oracle $U_{\psi}$ recognizing the state.  Aaronson's technique allows us to use pseudorandom states instead of truly random states with a negligible loss in security.

We note that, for quantum coins, it is not sufficient to choose a random binary string encoded randomly in a pair of non-orthogonal bases, such as the so-called ``BB84'' bases.  An adversary with a small number of quantum coins, say $O(\log n)$, can measure each qubit of the $O(\log n)$ tokens in both bases, and will with good probability find the correct basis choices and thus the random binary string, allowing her to then create arbitrarily many forged coins.

\subsection{Verification}\label{sec:qu:qm:bbcoins:verif}

Let $U_{\psi}$ be an oracle that recognizes the state $\ket{\psi}$ by flipping the sign of the phase of the state $\ket{\psi}$.  That is, $U_{\psi} \ket{\psi} = -\ket{\psi}$ and $U_{\psi} \ket{\phi} = \ket{\phi}$ for all $\ket{\phi}$ orthogonal to $\ket{\psi}$; in other words, $U_{\psi}=I-2\ket{\psi}\bra{\psi}$.

We can construct a verification circuit $V$ from the oracle $U_{\psi}$ as follows.  On the data readout register, input the state $\ket{0}$, then perform a Hadamard transformation on the ancilla.  Use the ancilla as the control bit of a controlled-$U_{\psi}$ applied to the input state $\rho$.  Then perform a Hadamard transformation again on the ancilla and measure it in the computational basis.  The circuit diagram is given in Figure~\ref{fig:bbcoins:verif-circ-example}.

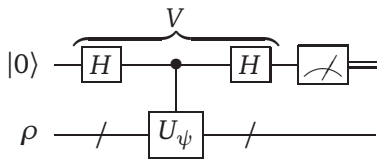
\begin{figure}[ht]
\[
\Qcircuit @C=1em @R=1em {
& & \mbox{$V$} & & & \\
\lstick{\ket{0}} & \gate{H} & \ctrl{1} & \gate{H} & \meter & \cw \\
\lstick{\rho} & {/} \qw & \gate{U_{\psi}} & {/} \qw & \qw & \qw 
\gategroup{2}{2}{3}{4}{0.5em}{^\}}
}
\]
\caption{Verification circuit for quantum coins $\ket{\psi}$ recognized using the oracle $U_{\psi}$.}
\label{fig:bbcoins:verif-circ-example}
\end{figure}

When a measurement in the computational basis is performed on the ancilla register, the result will be $\ket{1}$ when the input state $\rho$ is $\ket{\psi}$ and $\ket{0}$ when the input state is $\ket{\phi}$ for $\braket{\phi}{\psi}=0$.  Moreover, the state on the second register remains unchanged when its input is $\ket{\psi}$. 

The fact that a valid token is unchanged by the verification process  allows transferability of quantum coins.  When a customer spends a quantum coin at a store, the merchant, after verifying and accepting the coin, can retain the coin until the merchant needs to make change.  At that time, the merchant can give the coin to another user who, after optionally verifying the coin, can use that coin in another transaction.  (In fact, the verification process not only enables transferability but also enhances the robustness of the quantum coins.  Although over time a quantum state may decohere, at verification time the token may still be sufficiently close to the expected state $\ket{\psi}$ to pass the verification process with high probability.  If it does pass, then the measurement process will project the coin back into the original state $\ket{\psi}$.) 

\paragraph{Security.} 
The verification procedure described in the previous section yields a correct quantum money scheme: valid money tokens are  recognized.  We now discuss the security of such a scheme.  For unforgeability, we want that invalid tokens are recognized as being invalid and that it is difficult to forge new money.

\subsection{Black-box unforgeability}\label{sec:qu:qm:bbcoins:forgery}

To analyze the forgeability of the quantum coin scheme given in Figure~\ref{fig:bbcoins:verif-circ-example}, we suppose that the circuit for the unitary $U_{\psi}$ is a black box, meaning that no information can be obtained from observing its inner workings; equivalently, we assume that $U_{\psi}$ is given as an oracle.  Having made this assumption, we proceed to obtain a lower bound on the number of queries to the oracle that must be made in order to produce a state that has a particular overlap $p$ with $\ket{\psi}^{\otimes k+1}$, when the adversary is only given $k$ coins.  We show this result in the next section.

\begin{definition}\label{defn:qu:qm:bbcoins:unforgeable}
A quantum coin scheme $(V, \ket{\psi})$, where $\ket{\psi}$ is an $n$-qubit state, is \emph{black-box unforgeable} if, given an oracle $U_{\psi}$ recognizing the state $\ket{\psi}$ and $k$ copies of the state $\ket{\psi}$, for any $k \ge 0$, $k \in \poly(n)$, it is not possible for a quantum adversary using $\poly(n)$ queries to $U_{\psi}$ to produce a state $\rho$ such that $\bra{\psi}^{\otimes k+1} \rho \ket{\psi}^{\otimes k+1}$ is non-negligible.\footnote{In the language of Aaronson \cite{Aar09}, this is a single key private key quantum money scheme with completeness error 0 and soundness error negligible in $n$.}
\end{definition}

We note that our definition of unforgeability has the adversary producing a $(k+1)$-register state, each register of which should overlap well with $\ket{\psi}$.  An alternative formulation could be that the adversary needs to produce a multi-register state such that some $k+1$ of its registers, but not necessarily all of its registers, overlap well with $\ket{\psi}$.  These definitions are equivalent.  The adversary has access to a verification oracle and, for each of the many registers it constructs, could simply apply the verification oracle to each register and then trace out any registers that do not pass verification.  This requires additional calls to the verification oracle, but still only $\poly(n)$ calls to the oracle (since a polynomial-time adversary can only construct $\poly(n)$ registers), and hence remains within the constraints of the security argument above.

We note as well that it is not necessary to extend this definition to $k+\ell$ copies of $\ket{\psi}$: any adversary who can construct $k+\ell$ copies of $\ket{\psi}$ with non-negligible probability can in particular construct $k+1$ copies of $\ket{\psi}$ with non-negligible probability.  In other words, there are no ``long shots'' that pay off in expected value: the definition precludes being able to generate a very large number of coins with a very small probability but with non-negligible expected number of coins.

We now aim to show that a generic quantum coin scheme implemented with black-box oracles as in Figure~\ref{fig:bbcoins:verif-circ-example} is black-box unforgeable.  However, we cannot use the basic no-cloning theorem \cite{WZ82,Die82} or the result on approximate cloning \cite{BM07} because not only does a forger have copies of the state $\ket{\psi}$, the forger also has an oracle $U_{\psi}$ that will indicate whether the attempted cloning was successful.  Similarly, we cannot directly apply the $\Omega(\sqrt{N})$ lower bound on quantum search \cite{BBBV97} because the forger has not only an oracle $U_{\psi}$ recognizing the desired state but also some copies of the state itself.  Rather, we need a hybrid of these two results.

Aaronson \cite{Aar05} gives the following complexity-theoretic version of the no-cloning theorem that combines the lower bound for quantum search  with the no-cloning theorem.

\begin{theorem}[Theorem~5, \cite{Aar05}]\label{thm:qu:qm:bbcoins:forgery:aaronson}
Let $\ket{\psi}$ be an $n$-qubit pure state.  Suppose we are given the initial state $\ket{\psi}^{\tensor k}$ for some $k \ge 1$ as well as an oracle $U_{\psi}$ such that $U_{\psi} \ket{\psi} = - \ket{\psi}$ and $U_{\psi} \ket{\phi} = \ket{\phi}$ whenever $\braket{\phi}{\psi} = 0$.  Then to prepare a state $\rho$ such that 
\begin{equation}
\bra{\psi}^{k+1} \rho \ket{\psi}^{k+1} \ge p
\end{equation}
we need
\begin{equation}
\Omega \left( \frac{\sqrt{2^{n}p}}{k \log k} - k \right)
\end{equation}
queries to $U_{\psi}$.
\end{theorem}

This allows us to show that a quantum coin scheme is unforgeable in the black-box oracle model.

\begin{theorem}\label{thm:qu:qm:bbcoins:forgery}
Let $(V, \ket{\psi})$ be a quantum coin scheme, where $V$ is as in Figure~\ref{fig:bbcoins:verif-circ-example} with $U_{\psi}$ given as a black-box oracle, and $\ket{\psi}$ is an $n$-qubit pure state.  If not more than $\mathrm{poly}(n)$ coins are issued, then $(V, \ket{\psi})$ is black-box unforgeable.
\end{theorem}

\begin{proof}
Suppose otherwise.  Then there exists an adversary who, upon receiving $k$ copies of $\ket{\psi}$ and using $q = \poly(n)$ queries to $U_{\psi}$, can produce a state $\rho$ such that $\bra{\psi}^{\tensor k+1} \rho \ket{\psi}^{\tensor k+1} = p \in 1/\poly(n)$.  By Theorem~\ref{thm:qu:qm:bbcoins:forgery:aaronson}, we need
\begin{equation}
q = \Omega \left( \frac{\sqrt{2^{n}p}}{k \log k} - k \right) 
= \Omega \left( \frac{\sqrt{2^{n} /\poly(n)}}{\poly(n) \log \poly(n)} - \poly(n) \right) 
= \Omega \left( \frac{\sqrt{2^{n}}}{\poly(n)} \right) 
\end{equation}
queries to $U_{\psi}$.  But since the adversary is allowed only a polynomial number $q$ of queries to $U_{\psi}$, we have that $q \in \poly(n)$ and hence $\poly(n) = \Omega\left(\frac{\sqrt{2^{n}}}{\poly(n)} \right)$, which is a contradiction.  Thus the quantum coin scheme must be black-box unforgeable.
\end{proof}

\section{Quantum coins using blind quantum computation}\label{sec:qu:qm:blindcoins}

Blind quantum computation allows one party, Alice, to have another party, Bob, perform computations on her behalf without Bob learning any information about the input state, output state, or the operation performed.  

Blind quantum computation was first introduced by Childs \cite{Chi05} under the name ``secure assisted quantum communication''.  The basic idea is that Alice, who has limited quantum computational abilities (quantum communication, quantum storage, and controlled-$X$ and controlled-$Z$ gates) can have Bob securely perform arbitrary quantum computation, with quantum input and quantum output.  In Childs' protocol, Alice and Bob must perform large amounts of quantum communication, though this could be replaced by quantum teleportation (shared entanglement with Bell measurements and classical communication).

Broadbent, Fitzsimons, and Kashefi \cite{BFK09} present a protocol for blind quantum computation with quantum input and output using measurement-based quantum computation that needs only two rounds of quantum communication: one at the beginning and one at the end.  

Blind quantum computation could be used as follows for verifying quantum coins as follows.  The merchant, playing the role of Bob, implements the verification circuit blindly for the bank, playing the role of Alice.  The merchant receives the coin as the input to the circuit, and interacts with the bank who helps it implement the circuit.   In the \cite{BFK09} scheme, this requires mostly classical interaction, with round round of quantum interaction at the end for the final output correction.  In the end, the output state along with the accept/reject information is with the merchant.

Although the quantum communication requirements for verifying quantum coins using blind quantum computation are no better than simply teleporting the coin to the bank for verification, the quantum computation requirement for the bank is markedly reduced: instead of having to implement the full quantum circuit for coin verification for the thousands of coins being verified each second, it only has to perform step 5 of Protocol 3 of \cite{BFK09}, which consists of at most one $X$ gate and one $Z$ gate per coin qubit.

Obviously, it would be preferable to reduce this quantum communication requirement even further, for example by only requiring quantum communication at the beginning of the protocol and only classical communication for the remainder of the protocol, and without using shared entanglement for teleportation.  A protocol for doing so would be an interactive protocol for quantum circuit obfuscation, and quantum obfuscation is a long standing open problem (cf. \cite{Aar05b}).

\subsection*{Acknowledgements}

The authors gratefully acknowledge helpful discussions with Scott Aaronson, Anne Broadbent, Joseph Fitzsimons, Miklos Santha, and John Watrous.  M.M. was supported by Canada's NSERC, QuantumWorks, MITACS, CIFAR, CRC, ORF, the Government of Canada, and Ontario-MRI.  D.S. was supported by a Canada NSERC Postgraduate Scholarship and Sun Microsystems Laboratories.  Research performed while D.S. was at the University of Waterloo.

\small 

\bibliographystyle{halphads}
\bibliography{/Users/dstebila/Bibliography/Library}

\newcommand{\etalchar}[1]{$^{#1}$}
\begin{thebibliography}{BGI{\etalchar{+}}01b}
 \providecommand{\doi}[1]{{\sc doi}:\href{http://dx.doi.org/#1}{#1}}
 \providecommand{\urlprefix}{{\sc url} }
 \providecommand{\eprintprefix}{{\sc eprint} }

\bibitem[Aar05a]{Aar05}
Scott Aaronson.
\newblock Quantum copy-protection.
\newblock Private correspondence, 2005.

\bibitem[Aar05b]{Aar05b}
Scott Aaronson.
\newblock Ten semi-grand challenges for quantum computing theory, July 2005.
\newblock
  \urlprefix\url{http://www.scottaaronson.com/writings/qchallenge.html}.

\bibitem[Aar09]{Aar09}
Scott Aaronson.
\newblock Quantum copy-protection and quantum money.
\newblock In {\em IEEE 24th Conference on Computational Complexity (CCC) 2009}.
  IEEE, 2009.
\newblock \urlprefix\url{http://www.scottaaronson.com/papers/noclone-ccc.pdf}.
\newblock To appear.

\bibitem[AE07]{AE07}
Andris Ambainis and Joseph Emerson.
\newblock Quantum t-designs: t-wise independence in the quantum world.
\newblock In {\em Proc. 22nd Ann. {IEEE} Conference on Computational Complexity
  (CCC) 2007}, pp. 129--140. IEEE, June 2007.
\newblock \doi{10.1109/CCC.2007.26}.
\newblock
  \eprintprefix\href{http://arxiv.org/abs/quant-ph/0701126}{arXiv:quant-ph/070%
1126}.

\bibitem[AJK04]{AJK04}
Joseph~B. Altepeter, Daniel F.~V. James, and Paul~G. Kwiat.
\newblock 4 qubit quantum state tomography.
\newblock In Matteo Paris and Jaroslav {\v R}eh{\'a}{\v c}ek, editors, {\em
  Quantum State Estimation}, {\em Lecture Notes in Physics}, volume 649, pp.
  113--145. Springer, 2004.
\newblock \doi{10.1007/b98673}.

\bibitem[BBBV97]{BBBV97}
Charles~H. Bennett, Ethan Bernstein, Gilles Brassard, and Umesh Vazirani.
\newblock Strengths and weaknesses of quantum computing.
\newblock {\em SIAM Journal on Computing}, {\bf 26}(5):1510--1523, 1997.
\newblock \doi{10.1137/S0097539796300933}.
\newblock
  \eprintprefix\href{http://arxiv.org/abs/quant-ph/9701001}{arXiv:quant-ph/970%
1001}.

\bibitem[BBBW82]{BBBW82}
Charles~H. Bennett, Gilles Brassard, Seth Breidbard, and Stephen Wiesner.
\newblock Quantum cryptography, or unforgeable subway tokens.
\newblock In David Chaum, Ronald~L. Rivest, and Alan~T. Sherman, editors, {\em
  Advances in Cryptology -- Proc. {CRYPTO} '82}. Plenum Press, 1982.

\bibitem[BBF{\etalchar{+}}07]{BBFGT07}
Gilles Brassard, Anne Broadbent, Joseph Fitzsimons, S{\'e}bastien Gambs, and
  Alain Tapp.
\newblock Anonymous quantum communication.
\newblock In Kurosawa \cite{ASIACRYPT2007}, pp. 460--473.
\newblock \doi{10.1007/978-3-540-76900-2\_28}.
\newblock \eprintprefix\href{http://arxiv.org/abs/0706.2356}{arXiv:0706.2356}.

\bibitem[BBHT98]{BBHT98}
Michel Boyer, Gilles Brassard, Peter H{\o}yer, and Alain Tapp.
\newblock Tight bounds on quantum searching.
\newblock {\em Fortschritte der Physik}, {\bf 46}(4--5):493--505, 1998.
\newblock
  \doi{10.1002/(SICI)1521-3978(199806)46:4/5<493::AID-PROP493>3.0.CO;2-P}.
\newblock
  \eprintprefix\href{http://arxiv.org/abs/quant-ph/9605034}{arXiv:quant-ph/960%
5034}.

\bibitem[BFK09]{BFK09}
Anne Broadbent, Joseph Fitzsimons, and Elham Kashefi.
\newblock Universal blind quantum computation.
\newblock In {\em Proc. 50th Annual IEEE Symposium on Foundations of Computer
  Science (FOCS) 2009}. IEEE Press, 2009.
\newblock \eprintprefix\href{http://arxiv.org/abs/0807.4154}{arXiv:0807.4154}.
\newblock To appear.

\bibitem[BGI{\etalchar{+}}01a]{BGIRSVY01full}
Boaz Barak, Oded Goldreich, Russell Impagliazzo, Steven Rudich, Amit Sahai,
  Salil Vadhan, and Ke~Yang.
\newblock On the (im)possibility of obfuscating programs, 2001.
\newblock \eprintprefix\url{http://eprint.iacr.org/2001/069},
  \urlprefix\url{http://www.wisdom.weizmann.ac.il/~oded/p_obfuscate.html}.
\newblock Published as \cite{BGIRSVY01}.

\bibitem[BGI{\etalchar{+}}01b]{BGIRSVY01}
Boaz Barak, Oded Goldreich, Russell Impagliazzo, Steven Rudich, Amit Sahai,
  Salil Vadhan, and Ke~Yang.
\newblock On the (im)possibility of obfuscating programs.
\newblock In Joe Kilian, editor, {\em Advances in Cryptology -- Proc. {CRYPTO}
  2001}, {\em LNCS}, volume 2139, pp. 1--18. Springer, 2001.
\newblock \doi{10.1007/3-540-44647-8\_1}.
\newblock Full version available as \cite{BGIRSVY01full}.

\bibitem[BM07]{BM07}
Dagmar Bru{\ss} and Chiara Macchiavello.
\newblock Approximate quantum cloning.
\newblock In Dagmar Bru{\ss} and Gerd Leuchs, editors, {\em Lectures on Quantum
  Information}. Wiley-VCH, 2007.
\newblock \doi{10.1002/9783527618637}.

\bibitem[BT07]{BT07}
Anne Broadbent and Alain Tapp.
\newblock Information-theoretic security without an honest majority.
\newblock In Kurosawa \cite{ASIACRYPT2007}, pp. 410--426.
\newblock \doi{10.1007/978-3-540-76900-2\_25}.
\newblock \eprintprefix\href{http://arxiv.org/abs/0706.2010}{arXiv:0706.2010}.

\bibitem[CFN88]{CFN88}
David Chaum, Amos Fiat, and Moni Naor.
\newblock Untraceable electronic cash (extended abstract).
\newblock In Shafi Goldwasser, editor, {\em Advances in Cryptology -- Proc.
  {CRYPTO} '88}, {\em LNCS}, volume 403, pp. 319--327. Springer, 1988.
\newblock \doi{10.1007/0-387-34799-2\_25}.

\bibitem[Cha85]{Cha85}
David Chaum.
\newblock Security without identification: transaction systems to make big
  brother obsolete.
\newblock {\em Communications of the ACM}, {\bf 28}(10):1030--1044, October
  1985.
\newblock \doi{10.1145/4372.4373}.

\bibitem[Cha88]{Cha88}
David Chaum.
\newblock Privacy protected payments: Unconditional payer and/or payee
  untraceability.
\newblock In David Chaum and I.~Schaumuller-Bichl, editors, {\em Smartcard
  2000}, pp. 69--93. North Holland, 1988.

\bibitem[Chi05]{Chi05}
Andrew Childs.
\newblock Secure assisted quantum computation.
\newblock {\em Quantum Information and Computation}, {\bf 5}(6):456--466,
  September 2005.
\newblock
  \eprintprefix\href{http://arxiv.org/abs/quant-ph/0111046}{arXiv:quant-ph/011%
1046}, \urlprefix\url{http://www.rinton.net/xqic5/qic-5-6/456-466.pdf}.

\bibitem[CP92]{CP92}
David Chaum and Torben~Pryds Pedersen.
\newblock Transferred cash grows in size.
\newblock In Rainer~A. Rueppel, editor, {\em Advances in Cryptology -- Proc.
  {EUROCRYPT} '92}, {\em LNCS}, volume 658, pp. 390--407. Springer-Verlag,
  1992.
\newblock \doi{10.1007/3-540-47555-9\_32}.

\bibitem[Die82]{Die82}
D.~Dieks.
\newblock Communication by {EPR} devices.
\newblock {\em Physics Letters A}, {\bf 92}(6):271--272, November 1982.
\newblock \doi{10.1016/0375-9601(82)90084-6}.

\bibitem[Kur07]{ASIACRYPT2007}
Kaoru Kurosawa, editor.
\newblock {\em Advances in Cryptology -- Proc. {ASIACRYPT} 2007}, {\em LNCS},
  volume 4833. Springer, 2007.
\newblock \doi{10.1007/978-3-540-76900-2}.

\bibitem[MS06]{MS06}
Michele Mosca and Douglas Stebila.
\newblock Uncloneable quantum money.
\newblock In {\em Canadian Quantum Information Students' Conference (CQISC)
  2006}, Calgary, Alberta, August 2006.
\newblock
  \urlprefix\url{http://www.iqis.org/events/cqisc06/papers/Mon-1130-Stebila.pd%
f}.

\bibitem[MS07]{MS07}
Michele Mosca and Douglas Stebila.
\newblock A framework for quantum money.
\newblock In {\em Quantum Information Processing (QIP) 2007}, Brisbane,
  Australia, January 2007.

\bibitem[Ste09]{Ste09}
Douglas Stebila.
\newblock {\em Classical Authenticated Key Exchange and Quantum Cryptography}.
\newblock PhD thesis, University of Waterloo, 2009.
\newblock \eprintprefix\url{http://hdl.handle.net/10012/4295},
  \urlprefix\url{http://www.douglas.stebila.ca/research/papers/ste09/}.

\bibitem[TOI03]{TOI03}
Yuuki Tokunaga, Taisuaki Okamoto, and Nobuyuki Imoto.
\newblock Anonymous quantum cash.
\newblock In {\em ERATO Conference on Quantum Information Science (EQIS) 2003},
  September 2003.
\newblock
  \urlprefix\url{http://www.qci.jst.go.jp/eqis03/program/papers/O09-Tokunaga.p%
s.gz}.

\bibitem[Wee05]{Wee05}
Hoeteck Wee.
\newblock On obfuscating point functions.
\newblock In {\em Proc. 37th Annual ACM Symposium on the Theory of Computing
  (STOC)}, pp. 523--532. ACM Press, 2005.
\newblock \doi{10.1145/1060590.1060669}.
\newblock \eprintprefix\url{http://eprint.iacr.org/2005/001}.

\bibitem[Wie83]{Wie83}
Stephen Wiesner.
\newblock Conjugate coding.
\newblock {\em ACM SIGACT News}, {\bf 15}(1):78--88, 1983.
\newblock \doi{10.1145/1008908.1008920}.

\bibitem[WZ82]{WZ82}
William~K. Wootters and W.~H. Zurek.
\newblock A single quantum cannot be cloned.
\newblock {\em Nature}, {\bf 299}:802--803, October 1982.
\newblock \doi{10.1038/299802a0}.

\end{thebibliography}

\end{document}